\documentclass{article}

\usepackage[top=2.2cm,bottom=2cm,left=2cm,right=2cm,bindingoffset=0.5cm,headsep=10pt]{geometry}

\usepackage{algorithmic}
\usepackage{algorithm}
 \usepackage[caption=false,font=footnotesize]{subfig}
\usepackage[cmex10]{amsmath}
\allowdisplaybreaks[2]
\usepackage{amssymb}
\usepackage{amsthm}
\usepackage{floatflt}
\usepackage[caption=false,font=footnotesize]{subfig}
\usepackage{fixltx2e}

\usepackage{url}

\usepackage{graphicx}

\usepackage{xspace} 

\def\floor#1{\left\lfloor #1 \right\rfloor}
\def\ceiling#1{\left\lceil #1 \right\rceil}

\def\alg{{\sc Alg}\xspace}
\def\opt{{\sc Opt}\xspace}
\def\alc{{\sc AllClass}\xspace}
\def\sepc{{\sc SeparateClass}\xspace}
\newtheorem{theorem}{Theorem}

\newtheorem{lemma}[theorem]{Lemma}

\begin{document}

\title{Scheduling Light-trails on WDM Rings}

\author{Soumitra~Pal \qquad Abhiram~Ranade\\
   Department of Computer Science and Engineering, \\
   Indian Institute of Technology Bombay, \\
   Powai, Mumbai 400076, India. \\
   \texttt{\{mitra,ranade\}@cse.iitb.ac.in}
}
\date{}

\maketitle

\begin{abstract}
We consider the problem of scheduling communication on optical WDM
(wavelength division multiplexing) networks using the light-trails
technology.  We seek to design scheduling algorithms such that the given
transmission requests can be scheduled using minimum number of wavelengths
(optical channels).  We provide algorithms and close lower bounds for two
versions of the problem on an $n$ processor linear array/ring network.
In the {\em stationary} version, the pattern of transmissions (given)
is assumed to not change over time.  For this, a simple lower bound
is $c$, the congestion or the maximum total traffic required to pass
through any link.  We give an algorithm that schedules the transmissions
using $O(c+\log{n})$ wavelengths.  We also show a pattern for which
$\Omega(c+\log{n}/\log\log{n})$ wavelengths are needed.  In the {\em
on-line} version, the transmissions arrive and depart dynamically,
and must be scheduled without upsetting the previously scheduled
transmissions.  For this case we give an on-line algorithm which has
competitive ratio $\Theta(\log{n})$.  We show that this is optimal
in the sense that every on-line algorithm must have competitive ratio
$\Omega(\log{n})$.  We also give an algorithm that appears to do well
in simulation (for the classes of traffic we consider), but which has
competitive ratio between $\Omega(\log^2n/\log \log{n})$ and $O(\log^2n)$.
We present detailed simulations of both our algorithms.
\end{abstract}


\section{Introduction}

Light-trails~\cite{chlamtac2003light} are considered to be an attractive
solution to the problem of bandwidth provisioning in optical networks.
One key idea in this is the use of optical shutters which are inserted
into the optical fiber, and which can be configured to either let the
optical signal from one segment of the fiber pass through or block it from
being transmitted into the next segment.  By configuring some shutters
ON (signal let through) and some OFF (signal blocked), the network can
be partitioned into subnetworks in which multiple communications can
happen in parallel on the same light wavelength.  In order to use the
network efficiently, it is important to have algorithms for controlling
the shutters.  Notice that in the ON mode, light goes through a shutter
{\em without} being first converted to an electrical signal -- this is
one of the major advantages of the light-trail technology.

In this paper we consider the simplest scenario: two fiber optic rings,
one clockwise and one anticlockwise, passing through a set of some $n$
nodes, where typically $n<20$ because of technological considerations. At
each node of a ring there are optical shutters that can either be used
to block or forward the signal on each possible wavelength.  The optical
shutters are controlled by an auxiliary network (``out-of-band channel'').
It is to be noted that this auxiliary network is typically electronic, and the
shutter switching time is of the order of milliseconds as opposed to
optical signals which have frequencies of Gigahertz.

For this setting we give three algorithms for controlling the shutters,
or {\em bandwidth provisioning}.  Our first algorithm is for stationary
traffic, i.e., the communication demands between processors are known at
the beginning and do not change over time.  Our second and third algorithm
are for dynamic traffic, i.e., transmission requests arrive and depart
in an on-line manner, and a new request must be assigned light-trails
without requiring modifications (or with only minimal modification)
to light-trails created for currently active transmissions.  For both
problems, our objective is to minimize the number of wavelengths needed
to accommodate the given traffic.  Our results are applicable to the
setting in which a fixed number of wavelengths is available as follows.
If our analysis indicates that some $\lambda$ wavelengths are needed while
only $\lambda_0$ are available, then effectively the system will have to
be slowed down by a factor $\lambda/\lambda_0$.  This is of course only
one formulation; there could be other formulations which allow requests
to be dropped and analyze what fraction of requests are served.

The input to the stationary problem is a matrix $B$, in which $B(i,j)$
gives the bandwidth demanded between nodes $i$ and $j$, expressed as a
fraction of the bandwidth supported by a single wavelength.  We give an
algorithm which schedules this traffic using $O(c+\log n)$ wavelengths,
where $ c = \max_k \sum_{i,j|i\le k <j} B(i,j)$ is the maximum total
bandwidth demand, or the {\em congestion} at any link. The congestion as
defined above is a lower bound, and so our algorithm can be seen to use
a number of wavelengths close to the optimal.  The reader may wonder why
the additive $\log{n}$ term arises in the result.  We show that there
are communication matrices $B$ for which the congestion $c$ is small,
but which yet require $\Omega(c+\log{n}/\log\log{n})$ wavelengths.
In some sense, this justifies the form of our result.

For the on-line problem, we use the notion of competitive analysis
\cite{borodin1998oca,sleator1985amortized,karlin1988competitive}.
In this, an on-line algorithm which must respond without the knowledge of
the future is evaluated by comparing its performance to that of an {\em
off-line adversary}, an algorithm which is given all the transmission
requests at the beginning.  Clearly, the off-line adversary must perform
at least as well as the best on-line algorithm.  We establish that our
first algorithm is $\Theta(\log{n})$-competitive, i.e., it requires
$\Theta(\log{n})$ times as many wavelengths as needed by the off-line
adversary.  We also prove that no on-line algorithm can do better by
showing the lower bound on the competitive ratio of any algorithm for the
problem to be $\Omega(\log n)$.  A multiplicative $\Theta(\log{n})$ factor
might be considered to be too large to be relevant for practice; however,
the experience with on-line algorithm design is that such algorithms
often give good hints for designing practical algorithms.   We also give
a second algorithm for this problem, it is in fact a simplified version
of the first.  It actually performs better than the first algorithm in
many situations; however, we can prove that its competitive ratio is
worse, between $\Omega(\log^2 n/\log\log n)$ and $O(\log^2 n)$.

That brings us to our final contribution: we simulate two algorithms based
on our on-line algorithms for some traffic models.  We compare them to
a baseline algorithm which keeps the optical shutter switched OFF only
in one node for each wavelength. Note that at least one node should
switch OFF its optical shutter otherwise light signal will interfere
with itself after traversing around the ring. We find that except for
the case of very low traffic, our algorithms are better than the baseline.
For very local traffic, our algorithms are in fact much superior.

The rest of the paper is organized as follows. We begin in
Section~\ref{sec:prev_work} by comparing our work with previous
related work.  Section~\ref{sec:snap} discusses our algorithm for the
stationary problem.  Section~\ref{sec:lb} gives an example instance of
the stationary problem where the congestion lower bound is weak.  In
Section~\ref{sec:online} we describe our two algorithms for the on-line
problem.  In Section~\ref{sec:problemlb} we show that every on-line
algorithm must have competitive ratio  $\Omega(\log{n})$.  In
Section~\ref{sec:simul} we give results of simulation of our on-line
algorithms.

\section{Previous Work}
\label{sec:prev_work}

After the light-trail technology was was introduced in
\cite{chlamtac2003light}, a variety of hardware models have emerged.
For example, \cite{gumaste2003mil} has a mesh implementation of
light-trails for general networks. The paper \cite{gumaste2004optimizing}
implements a tree-shaped variant of light-trail, called as clustered
light-trail, for general networks. The paper \cite{ye2005traffic}
describes `tunable light-trail' in which the hardware at the beginning
works just like a simple light-path but can be tuned later to act as a
light-trail. There is some preliminary work on multi-hop light-trails
\cite{multihop} in which transmissions are allowed to go through
a sequence of overlapping light-trails.  Survivability in case of
failures is considered in~\cite{balasubramanian2005ltn} by assigning
each transmission request to two disjoint light-trails.

A variety of performance objectives have been proposed.  Several
objectives are mentioned in the seminal paper \cite{gumaste2004lto}
-- to minimize total number of light-trails used, to minimize queuing
delay, to maximize network utilization etc.  Most of the work in the
literature seems to solve the problem by minimizing total number
of light-trails used \cite{fang:olt, gumaste2007hao, ayad2007eoa,
wu2006opn03}.  Though the paper \cite{gumaste2007hao} suggests that
minimizing total number of light-trails also minimizes total number
of wavelengths, it may not be always true. For example, consider a
transmission matrix in which $B(1,2)=B(3,4)=0.5$ and $B(2,3)=1$. To
minimize total number of light-trails used, we create two light-trails
on two different wavelengths. Transmission $(2,3)$ is put in one
light-trail and transmissions $(1,2)$ and $(3,4)$ are put in the other
light-trail. On the other hand, to minimize total number of wavelengths,
we put each of them in a separate light-trail on a single wavelength.
We believe that minimizing the number of light-trails is motivated by
the goal of minimizing the book-keeping and the scheduler overhead.
However, we do not think this can be more important than reducing the
number of wavelengths needed (or reducing the slowdown the system will
face if the number of wavelengths is fixed).  There are few other models
as well, e.g.  \cite{balasubramanian2005network} minimizes total number
of transmitters and receivers used in all light-trails.

The general approach followed in the literature to solve the
stationary problem is to formulate the problem as an integer linear
program (ILP) and then to solve the ILP using standard solvers. The
papers \cite{fang:olt,gumaste2007hao} give two different ILP
formulations. However, solving these ILP formulations takes prohibitive
time even with moderate problem size since the problem is NP-hard.
To reduce the time to solve the ILP, the paper~\cite{ayad2007eoa}
removed some redundant constraints from the formulation and added
some valid-inequalities to reduce the search space. However, the ILP
formulation still remains difficult to solve.

Heuristics have also been used. The paper~\cite{ayad2007eoa} solves
the problem in a general network. It first enumerates all possible
light-trails of length not exceeding a given limit. Then it creates a list
of eligible light-trails for each transmission and a list of eligible
transmissions for each light-trail.  Transmissions are allocated in an
order combining descending order of bandwidth requirement and ascending
order of number of eligible light-trails. Among the eligible light-trails
for a transmission, the one with higher number of eligible transmissions
and higher number of already allocated transmissions is given preference.
The paper~\cite{wu2006opn03} used another heuristic for the problem in
a general network. For a ring network, \cite{gumaste2007hao} used
three heuristics.

For the problem on a general network, \cite{balasubramanian2005ltn} solves
two subproblems. The first subproblem considers all possible light-trails
on all the available wavelengths as bins and packs the transmissions
into compatible bins with the objective of minimizing total number of
light-trails used. The second subproblem assigns these light-trails to
wavelengths. The first subproblem is solved using three heuristics and
the second problem is solved by converting it to a graph coloring problem
where each node corresponds to a light-trail and there is an edge between
two nodes if the corresponding light-trails conflict with each other.

For the on-line problem, a number of models are possible.  From the
point of view of the light-trail scheduler, it is best if transmissions
are not moved from one light-trail to another during execution, which
is the model we use.  It is also appropriate to allow transmissions to
be moved, with some penalty.  This is the model considered
in~\cite{gumaste2007hao, lodha2007sol}, where the goal is to minimize
the penalty, measured as the number of light-trails constructed.  The
distributions of the transmissions that arrive is also another
interesting issue.  It is appropriate to assume that the distribution
is fixed, as has been considered in many simulation studies including
our own.  For our theoretical results, however, we assume that the
transmission sequence can be arbitrary.  The work in~\cite{gumaste2007hao}
assumes that the traffic is an unknown but gradually changing
distribution.  It uses a stochastic optimization based heuristic
which is validated using simulations.  The paper \cite{ayad2007eoa}
considers a model in which transmissions arrive but do not depart.
Multi-hop problems have also been considered,
e.g.~\cite{zhang2005dynamic}.  An innovative idea to assign
transmissions to light-trails using {\em on-line auctions} has been
considered in~\cite{gumaste2006daa}.

Our problem as formulated is in fact similar to the
problem of scheduling communications on reconfigurable
bus architectures~\cite{DBLP:journals/ijcsa/WankarA09,
bondalapati1997reconfigurable, nakano1995bibliography}. Many
models of reconfigurable bus architectures have been proposed
and studied -- Reconfigurable Networks (RN)~\cite{ben1991power},
Bus Automation~\cite{rothstein1988bus}, Configurable Highly
Parallel Computer (CHiP)~\cite{snyder1982introduction}, Content
Addressable Array Parallel Processor (CAAPP)~\cite{weems1989image},
Reconfigurable Mesh (RMESH)~\cite{hao2002selection}, Reconfigurable
Buses with Shift Switching (REBSIS)~\cite{lin1995reconfigurable},
Reconfigurable Multiple Bus Machine (RMBM)~\cite{trahan1996power},
Distributed Memory Bus Computer (DMBC)~\cite{sahni1995data},
Mesh With Reconfigurable Bus ($M_r$)~\cite{rajasekaran1993mesh},
Polymorphic Processor Array (PPA)~\cite{maresca1993polymorphic},
Polymorphic Torus Network~\cite{li1989polymorphic}, Processor Array
with Reconfigurable Bus System (PARBS or PARBUS)~\cite{jang1995optimal}
and others.  In these models we have a graph in which processors are
vertices and edges are communication links, however, a processor can
choose to electrically connect (or keep separate) the communication
links incident to it.  If links are connected together (like setting
the shutter ON), the communication goes through (as well as being
read by the processor).  In this way the entire network can be made
to behave like few long or many short buses, as per the need of
the application running on the network.  Models using optical buses
have also been proposed -- Optical Communication Parallel Computer
(OCPC)~\cite{geréb1992efficient}, Array with Reconfigurable Optical Bus
(AROB)~\cite{pavel1996matrix}, Linear Array with Reconfigurable Pipelined
Bus System (LARPBS)~\cite{pan1998linear} and so on. Though these models
use optical signals instead of electrical signals for communication,
at an abstract level they are equivalent with PARBS and its variants. So
we use the generic term \emph{reconfigurable bus system} to denote all
the models mentioned above.

At an abstract level, the reconfigurable bus system is similar to our
light-trail model, as both models use controllable switches to dynamically
reconfigure a bus into multiple subbuses.  In both models, changing the
state of the switch takes very long as compared to the data rates on the
buses.  However, typically, reconfigurable bus systems have only one bus,
rather than allowing multiple wavelengths like the light-trail model.
A second difference is in the context in which the two models have been
studied.  The light-trail model has been studied more by the optical
network community, and the focus has been how to schedule relatively long
duration communication requests (connection based) without having any
graphical regularity.  Reconfigurable bus systems have been studied more
in the context of parallel computing, and the analyses have been more of
entire algorithms running on them.  These analyses typically concern short
messages and the communication patterns are often regular, such as those
arising in finding maximum/OR/XOR of numbers~\cite{miller1993parallel,
li1989polymorphic}, matrix multiplication~\cite{li1999parallel},
prefix computation~\cite{miller1993parallel, nakano1995prefix},
problems on graphs~\cite{subbaraman1993list, wang2007efficient},
sorting~\cite{pan1998efficient} and so on.  PRAM simulation on
reconfigurable bus~\cite{wang1990two,li2000efficient}, particularly in
the case of randomized assignment of shared memory cells, generates
random communication patterns. However, because these patterns are
drawn from a uniform distribution, they end up being quite regular
(and much of the analysis is to find regular patterns that are
supersets of what is required).  So even this work does not consider
truly arbitrary/irregular patterns which are our prime interest, for
the on-line as well as off-line (stationary) scenarios.\footnote{It is
interesting to note that the communication patterns for PRAM simulation
are uniformly random across the network because the PRAM address space
is {\em hashed}, i.e., distributed randomly.  Hashing has the effect
of converting possibly local communication going a short distance to
a random communication which most likely goes long distance.  Such a
strategy is inherently wasteful in utilization of bandwidth.  It seems
much better to directly deal with the arbitrary communication pattern
which arises in PRAM simulation in the first place.}

\subsection{Remarks}

As may be seen, there are a number of dimensions along which the work in
the literature may be classified: the network configuration, the kind of
problems attempted, and the solution approach.  Network configurations
starting from simple linear array/rings~\cite{pan1998efficient,
gumaste2007hao, lodha2007sol} to full structured/unstructured
networks~\cite{subbaraman1993list, balasubramanian2005ltn, fang:olt,
ayad2007eoa, wu2006opn03, zhang2005dynamic} have been considered, both
in the optical communication literature as well as the reconfigurable
bus literature.  The stationary problem as well as the dynamic problem has
been considered, with additional minor variations in the models.  Finally,
three solution approaches can be identified.  First is the approach
in which scheduling is done using exact solutions of Integer Linear
Programs \cite{fang:olt, gumaste2007hao, ayad2007eoa}.  This is useful
for very small problems.  For larger problems, using the second approach,
a variety of heuristics have been used \cite{balasubramanian2005ltn,
gumaste2007hao, ayad2007eoa, wu2006opn03}.  The evaluation of the
scheduling algorithms has been done primarily using simulations. The
third approach could be theoretical. However, except for some work
related to random communication patterns~\cite{rajasekaran1993mesh,
suel1994routing, rajasekaran1997sorting}, we see no theoretical analysis
of the performance of the scheduling algorithms.

In contrast, our main contribution is theoretical.  We give algorithms
with {\em provable} bounds on performance, both for the stationary
and the on-line case. Our work uses the competitive analysis approach
\cite{borodin1998oca} for the on-line problem. We use techniques
of approximation algorithms to solve the stationary problem. To our
knowledge, this competitive analysis and approximation algorithm approach
to solve the light-trail scheduling problems has not been used in the
literature. We also give simulation results for the on-line algorithms.

\section{The Stationary Problem}
\label{sec:snap}

In this section, instead of considering two unidirectional rings,
we consider a linear array of $n$ nodes, numbered~0 to $n-1$. The
link between the two consecutive nodes $i$ and $i+1$ is numbered $i$.
Communication is considered undirected.  This simplifies the discussion;
it should be immediately obvious that all results directly carry over
to two directed rings mentioned in the introduction.

In WDM, the physical optic fiber carrying signals of $k$ different
wavelengths from node $0$ through node $n-1$ is logically thought of as
$k$ independent parallel fibers each carrying signals of a single
wavelength. On a light-trail based WDM network, each of these logical
fibers is broken into segments by suitably setting the shutter at each
node. For each segment, only the nodes at the endpoints have their
shutters OFF; other nodes have their shutters ON. Note that each node
has a separate shutter for each wavelength and the shutters for all
wavelengths at nodes $0$ and $n-1$ are always OFF. The segment between
two OFF shutters is a light-trail. A transmission from $i$ to $j$ can
be assigned to a light-trail $L$ only if $u\le i<j \le v$ where $u,v$
are the end nodes of the light-trail. Further the sum of the required
bandwidths of all transmissions assigned to any single light-trail must
not exceed the capacity of a wavelength.

The input for the stationary problem is a matrix $B$ with $B(i,j)$
denoting the bandwidth requirement for the transmission from node $i$
to node $j$, without loss of generality, as a fraction of the bandwidth
capacity of a single wavelength. Capacity of each wavelength is $1$. The
goal is to schedule these in minimum number of wavelengths $w$. The
output must give $w$ as well as the light-trails used on each wavelength
and the mapping of each transmission to a light-trail that serves it.

It will be convenient to represent/visualize schedules geometrically.  We
will use the $x$ axis to represent our processor array, with processors at
integer points and the $y$ axis to represent the wavelengths numbered~$0,
1, 2,$ and so on.  The region bordered by $y=k$ and $y=k+1$ will be used
to depict the transmissions assigned to the wavelength numbered~$k$.
The region will be partitioned with vertical lines at the nodes where
the shutters are OFF. Each of the rectangular parts in the partition
represents a light-trail created on the corresponding wavelength. A
transmission from node $i$ to node $j$ having bandwidth requirement $b$
will be denoted as $[i,j]$ and represented as a rectangle of length $j-i$
and height $b$ located horizontally in the region between $x=i$ and
$x=j$ and vertically within the region corresponding to the light-trail
in which it is scheduled.  We will also use the terms \emph{length},
\emph{extent}, and \emph{height} of transmission $[i,j]$ to mean $j-i$,
the interval $[i,j]$, and $b$ respectively. Unless there is ambiguity in
the context we will also use $[i,j]$ to denote a light-trail with end
nodes $i$ and $j$. Similarly we will also use the terms \emph{length},
and \emph{extent} of light-trail $[i,j]$ to mean $j-i$ and the interval
$[i,j]$ respectively. All light-trails have height $1$.

Fig.~\ref{fig:ltrpinst} shows possible light-trail configurations and
optimum solutions for an example instance of the stationary problem where
the transmission matrix $B$ has non-zero entries $B(0,1)= B(1,2)=0.6,
B(0,2)=0.4$ only. There are two possible ways of creating light-trails
on a wavelength, depending on whether the shutter at node $1$ is set ON
or OFF. Note that the shutters at node $0$ and $2$ are always OFF. In
part (a) we show the first case where one wavelength (numbered~$0$)
has only one light-trail $[0,2]$, i.e., the shutter at node $1$ is ON.
Transmissions $[0,1]$ and $[0,2]$ are assigned to the light-trail. This
assignment is valid because the total bandwidth requirement of the
two transmissions does not exceed $1$.  However we can not also assign
transmission $[1,2]$ to the same light-trail because the total bandwidth
requirement of all three transmissions exceeds $1$.  So we need another
wavelength (numbered~$1$) to assign the transmission $[1,2]$.  The second
case is shown in part (b) where an wavelength (numbered~$0$) is divided
into two light-trails $[0,1]$ and $[1,2]$ by setting the shutter at
node $1$ OFF. The assignment of the transmissions to the light-trails,
as shown, is valid because for each of the two light-trails the total
bandwidth requirement of the transmission assigned in it does not exceed
$1$. However, we can not assign the transmission $[0,2]$ to either of the
light-trails because a transmission can not go across a OFF shutter. So we
need another wavelength (numbered~$1$) with a single light-trail $[0,2]$
to assign the transmission $[0,2]$. In both cases, we need at least $2$
wavelengths. Hence the optimal solution takes $2$ wavelengths.

\begin{figure}[h!tb]
\centering
\includegraphics[width=0.9\textwidth]{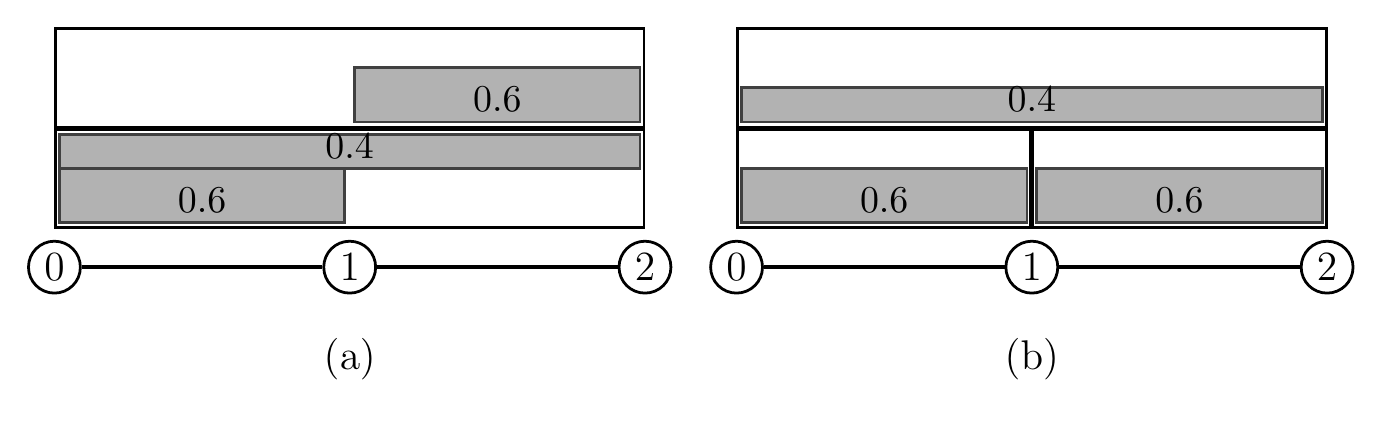}
\caption{Two optimum solutions of an instance of the stationary problem on
a linear array network with $n=3$ nodes}
\label{fig:ltrpinst}
\end{figure}

It is customary to consider two problem variations: {\em non-splittable},
in which a transmission must be assigned to a single light-trail, and
{\em splittable}, in which a transmission can be split into two or more
transmissions by dividing up the bandwidth requirement, and each of them
can be assigned to different light-trails. Note that when a transmission
is split into multiple transmissions, the length and extent remains same,
only the height is divided. Our results hold for both variations.

Note that the Bin Packing problem which is NP-hard, is a special case
of the stationary problem where each item corresponds to a transmission
from node $0$ to node $n-1$ and each bin corresponds to a light-trail
(and to a wavelength too because each light-trail completely occupies
a wavelength).  Thus the non-splittable stationary problem is NP-hard.
We do not know whether the splittable problem is also NP-hard.

We will use $c_l(S)$ to denote the congestion induced on a link $l$
by a set $S$ of transmissions.  This is simply the total bandwidth
requirement of those transmissions from $S$ requiring to cross link $l$.
Clearly $c(S)=\max_l c_l(S)$, the maximum congestion over all links,
is a lower bound on the number of wavelengths needed.  Finally if $t$ is
a transmission, then we abuse notation to write $c_l(t),c(t)$, instead
of $c_l(\{t\}),c(\{t\})$, for the congestion contributed by $t$ only,
which is equal to the bandwidth requirement of $t$. Let $R$ be the set
of all transmissions of an instance of the stationary problem. We will
use $c$ to denote the overall congestion $c(R)$.

\subsection{Algorithm Overview}

Getting an algorithm which requires only $O(c\log n)$ wavelengths is easy.
If $c_{n/2}$ denotes the congestion of the link between node $n/2$ and
node $n/2+1$, then the transmissions crossing this link can be scheduled
in $\ceiling{c}\ge \ceiling{c_{n/2}}$ wavelengths for the splittable
case, and twice that many for the non-splittable case (using ideas from
Bin Packing~\cite{coffmanjr1996aab}).  The remaining transmissions do
not cross the middle link, and hence can be scheduled by separately
solving two subproblems, one for the transmissions on each half of the
array. The two subproblems can share the wavelengths.  If $\lambda(n,c)$
denotes the number of wavelengths used for scheduling transmissions
of congestion at most $c$ in a linear array of $n$ nodes, we have the
recurrence $\lambda(n,c)=O(\ceiling{c})+\lambda(c,n/2)$.  Clearly this
solves to $O(c\log{n})$.  Getting a solution using $O(c+\log{n})$
wavelengths requires a different divide and conquer approach.

The first observation behind our algorithm is that it is relatively
easy to get a good schedule if all the transmissions had the same length
(see Section~\ref{ssec:conquer}).  So we  divide the transmissions into
classes based on their lengths, then schedule each class separately and
finally merge the schedules.  Naively scheduling and merging each class
would give us an $O(c\log{n})$ wavelength algorithm; with some care we do
get the required $O(c+\log{n})$ wavelength algorithm.  Our algorithm is:
\begin{enumerate}
\item {\em Partition into classes.} Say a transmission belongs to
 class $i$ if its length is between $2^{i-1}$ (exclusive) and $2^i$
 (inclusive).  Let $R_i$ denote the set of transmissions of class $i$,
 for $i=0$ to $\ceiling{\log_2(n-1)}$.
\item {\em Schedule transmissions of each class separately.} It will
 be seen that each class can be scheduled efficiently, i.e. using
 $O(1+c(R_i))$ wavelengths.
\item {\em Merge the schedules of different classes.} We do not simply
 collect together the schedules constructed for the different
 classes, but do need to mix them together, as will be seen.
\end{enumerate}

Scheduling classes $R_0,R_1$ is easy.  Note that each transmission in
$R_0$ has length $1$. So they can be assigned to light-trails created
by simply putting shutters OFF at every node on all the wavelengths that
are to be used.   Now for a fixed $l$ consider the light-trails $[l,l+1]$
on all the wavelengths. Each of these light-trails can be thought of as
a bin in which the transmissions $[l,l+1]$ are to be assigned. Clearly,
$\ceiling{c_l(R_0)}$ light-trails will suffice for the splittable case,
and twice that many for the non-splittable case (using ideas from Bin
Packing~\cite{coffmanjr1996aab}). Since the light-trails for different
$l$ do not overlap, they can be on the same wavelength. So $\max_l
O(\ceiling{c_l(R_0)}) = O(\ceiling{c(R_0)})$ wavelengths will suffice.
Transmissions in $R_1$ have length $2$. So they can be assigned
to light-trails created on two sets of wavelengths -- one having
shutters OFF at even nodes and the other having shutters OFF at odd
nodes. Transmissions starting at an even (odd) node are assigned to a
light-trail on a wavelength of the first (second) set. Using a argument
similar for the transmissions in $R_0$, we can show that each of these
sets require $O(\ceiling{c(R_1)})$ wavelengths.  So for the rest of this
paper we only consider classes $2$ and larger.

\subsection{Schedule Class $i\ge 2$}
\label{ssec:conquer}

It seems reasonable that if the class $R_i$ is further split into
subclasses each of which has $O(1)$ congestion, then the subclasses could
be scheduled using $O(1)$ wavelengths.  This intuition is incorrect for
an arbitrary collection of transmissions with congestion $O(1)$, as will be
seen in Section~\ref{sec:lb}.  However, the intuition is correct when
the transmissions have nearly the same length, as they do when taken from
any single $R_i$.

\begin{lemma}
There exists a polynomial time procedure to partition $R_i$ into
sets $S_1,S_2,\ldots, S_k$ where $k \le \ceiling{c(R_i)}$  such that
\emph{(i)}~$c(S_j) < 4$ for all $j$, and \emph{(ii)}~if a transmission in
$S_j$ uses link $l$ then $\ceiling{c_l(R_i)}  \ge j$.  \label{lem:breakup}
\end{lemma}

\begin{proof}
We start with $T_1=R_i$, and in general given $T_j$ we pick a subset
of transmission $S_j$ from $T_j$ using a procedure described below and
repeat with the remaining transmissions $T_{j+1}=T_j \setminus S_j$
until $T_{j+1}$ becomes empty for some value $k$ of $j$.

For each link $l$ from left to right, we greedily pick transmissions
crossing link $l$ into $S_j$ until we have removed at least unit
congestion from $c_l(T_j)$ or reduced $c_l(T_j)$ to 0. Note that if the
transmissions already picked while considering the links on the left of
$l$ also have congestion at least $1$ at link $l$ then we do not add
any more transmission while considering link $l$. So at the end the
following condition holds:
\begin{equation}
\forall l,\, c_l(S_j) \left\{ \begin{array}{ll}
=c_l(T_j) & \text{if} \; c_l(T_j)\le 1, \; \text{and} \\
\ge 1 & \text{otherwise.} \end{array} \right.
\label{eq:ensure_min}
\end{equation}
However, to make sure that $c(S_j)$ is not large, we move back
transmissions from $S_j$, in the reverse order as they were added, into
$T_j$ so long as condition \eqref{eq:ensure_min} remains satisfied. It can
be seen that the construction of $S_j$ takes at most $n|T_j|$ time in the
pick-up step and also in the move-back step.

Now we show that condition (i) of the lemma is satisfied, i.e., $c(S_j)
< 4$ for all $j$. At the end of move-back step, for any transmission
$t \in S_j$ there must exist a link $l$ such that $c_l(S_j)<1+c(t)$
otherwise $t$ would have been removed. We call $l$ as a \emph{sweet spot}
for $t$. Since $c(t)\le 1$ we have $c_l(S_j)<2$ for any sweet spot $l$.

Now consider any link $x$.  Of the transmissions through $x$, let
$L_1$ ($L_2$) denote transmissions having their sweet spot on the
left (right) of $x$.  Consider $y$, the rightmost of these sweet
spots of some transmission $t\in L_1$.  Note first that $c_y(S_j)<
2$. Also all transmissions in $L_1$ pass through both $x,y$.
Thus $c_x(L_1) = c_y(L_1) \leq c_y(S_j)<2$.  Similarly, $c_x(L_2)<2$.
Thus $c_x(S_j)=c_x(L_1)+c_x(L_2)<4$. But since this applies to all links
$x$, $c(S_j)<4$.

To show that the condition (ii) is also satisfied, suppose $S_j$
contains a transmission that uses some link $l$. The construction
process above must have removed at least unit congestion from $l$
in every previous step $1$ through $j-1$.  Thus $c_l(R_i) > j-1$.
That implies $\ceiling{c_l(R_i)} \ge j$.  This also implies that $k \le
\max_l \ceiling{c_l(R_i)} = \ceiling{c(R_i)}$.
\end{proof}

A transmission $t$ is said to cross a node $u$ if $t$ starts at a node
on the left of $u$ and ends at a node on the right of $u$. Since every
transmission $t$ in $S_j$ has length at least $2^{i-1}+1$, $t$ must cross
some node whose number is a multiple of $2^{i-1}$.  The smallest numbered
such node is called the \emph{anchor} of $t$. The \emph{trail-point}
of a transmission $t$ is the right most node numbered with a multiple of
$2^{i-1}$ that is on the left of the anchor of $t$.  If the transmission
has trail-point at node $q2^{i-1}$ for some $q$, then we define $q\bmod 4$
as its {\em phase}.

\begin{lemma}
The set $S_j$ can be scheduled using $O(1)$ wavelengths.
\label{lem:construct}
\end{lemma}

\begin{proof}
We partition $S_j$ further into sets $S_j^p$ containing transmissions
of phase $p$.  Note that the transmissions in any $S_j^p$ either
overlap at their anchors, or do not overlap at all. This is because
if two transmissions in $S_j^p$ have different anchors, then these two
anchors are at least $2^{i+1}$ distance apart.  Since the length of each
transmission is at most $2^i$, the two transmissions can not intersect.

So for the set $S_j^p$, consider $4$ wavelengths, each having shutters
OFF at nodes numbered $(4q+p)2^{i-1}$. Let $x=(4q+p)2^{i-1}$ and
$y=(4(q+1)+p)2^{i-1} = x+2^{i+1}$ be two nearest node having shutters
OFF. Among the light-trails thus created, for a fixed $q$, each of the $4$
light-trails $[x,y]$ can be thought of as a bin in which the transmissions
having extent totally within $[x,y]$ and total bandwidth requirement
at most 1 are to be assigned. This is an instance of the Bin Packing
problem. Clearly, for a fixed $q$, these $4$ light-trails will suffice
for the splittable case, because $c(S_j^p)<4$. Since the light-trails for
different $q$ do not overlap, the instances of the Bin Packing problem
can share wavelengths and hence these $4$ wavelengths will suffice.
For the non-splittable case, $8$ wavelengths will suffice, using standard
Bin Packing ideas, e.g., First-Fit~\cite{coffmanjr1996aab}.

Thus all of $S_j$ can be accommodated in at most $16$ wavelengths for the
splittable case, and at most $32$ wavelengths for the non-splittable case.
\end{proof}

\begin{lemma}
The entire set $R_i$ can be scheduled such that at each link $l$ there are
$O(C_l(R_i)+1)$ light-trails.
\label{lem:trim}
\end{lemma}

\begin{proof} We first consider the light-trails as constructed in
Lemma~\ref{lem:construct}. In this construction, uniformly at all links
there are at most $\ceiling{c(R_i)} \le c(R_i)+1$ sets of light-trails
such that each set corresponds to $O(1)$ light-trails created to schedule
the transmissions of an $S_j$. Note that $c(R_i) = \max_l c_l(R_i)$. So,
in this construction the condition of the lemma is surely satisfied for
the link where the congestion is maximum. For other links the condition of
the lemma may not be satisfied because (1) there may be empty light-trails
and (2) some light-trails may contain links that are not used by any of
the transmissions associated with the light-trail. So we remove empty
light-trails and in case (2) we shrink the light-trails by removing
the unused links (which can only be near either end of the light-trail
because all transmissions assigned to a light-trail overlap at their
anchor). We prove next that with this modification, the condition of
the lemma is satisfied.

Let $j$ be the largest such that a transmission from $S_j$ uses link
$l$. After the modification the light-trails that carries transmissions
from $S_{j'}$ for $j'>j$ do not use link $l$. So now there are $j$
sets of light-trails using link $l$ such that each set has $O(1)$
light-trails. However we know from Lemma~\ref{lem:breakup} that $j
\le \ceiling{c_l(R_i)} \le c_l(R_i)+1$. Thus there are a total of
$O(j)=O(c_l(R_i)+1)$ light-trails at link $l$.
\end{proof}

\subsection{Merge Schedules of All Classes}
\label{sec:all_class}

If we simply collect together the wavelengths as allocated above, we
would get a bound $O(c\log n)$.  Note however, that if two light-trails,
one for transmissions in class $i$ and the other for transmissions in
class $j$, are spatially disjoint, then they could possibly share the
same wavelength. Given below is a systematic way of doing this, which
gets us a sharper bound.

\begin{theorem}
The entire set $R$ can be scheduled using $O(c+\log n)$ wavelengths.
\end{theorem}

\begin{proof}
We know that after the modification in Lemma~\ref{lem:trim}, at each link
$l$ there are a total of $O(c_l(R_i)+1)$ light-trails for each class $i$.
Thus summing over all classes, the total number of light-trails at $l$
are $O(c_l(R)+\log n)$.

Think of each light-trail as an interval, giving us a collection
of intervals such that any link $l$ has at most $O(c_l(R)+\log
n)=O(c+\log{n})$ intervals.  Now this collection of intervals can be
colored using $O(c+\log{n})$ colors~\cite{olariu1991optimal}.  Now for
each color $w$, use a separate wavelength and configure the light-trails
corresponding to the intervals of color $w$ by setting the shutters
OFF at the nodes corresponding to the endpoints of the intervals. Hence
$O(c+\log n)$ wavelengths suffice.
\end{proof}

\section{On the Congestion Lower Bound}
\label{sec:lb}

We now consider an instance of the stationary problem. For convenience,
we assume there are $n+1$ nodes numbered $0,\ldots, n$ where $n=2^k$ for
some $k$ and all logarithms are with base $2$. All the transmissions have
same bandwidth requirement $b=1/(\log n+1)$.

First, we have a transmission going from $0$ to $n$.  Then a
transmission from $0$ to $n/2$ and a transmission from $n/2$ to $n$.
Then four transmissions spanning one-fourth the distance, and so on.
Thus we have transmissions of $\log n+1$ classes, each class having
transmissions of same length. In class $i \in \{0,1,\ldots,\log n\}$ there
are $n/2^i$ transmissions $B(s_{ij},d_{ij})=b$ where $s_{ij}=j2^i,
d_{ij}=(j+1)2^i$ for all $j=0,1,\ldots,(n/2^i)-1$. All other entries of
$B$ are $0$. This is illustrated in Fig.~\ref{fig:lbpattern}(a) for $n=17$.

\begin{figure*}[h!tb]
\centering
\includegraphics[width=\textwidth]{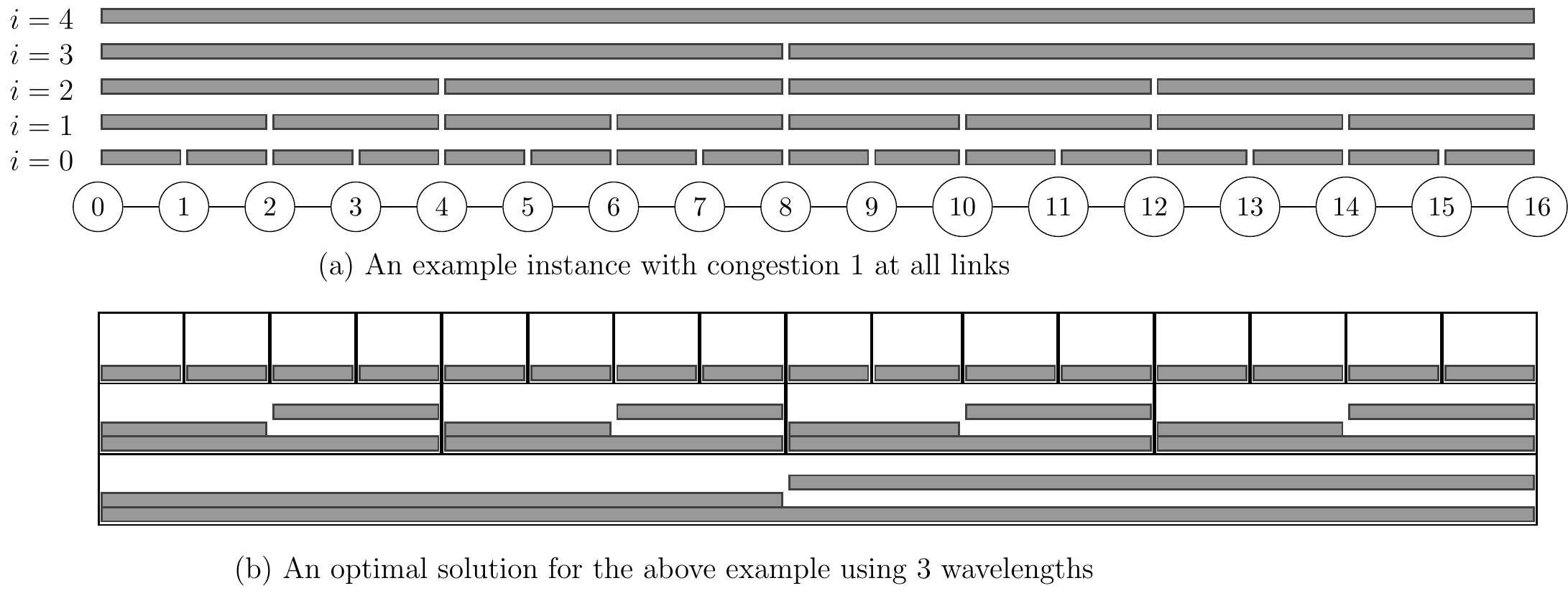}
\caption{An example instance where congestion bound is weak}
\label{fig:lbpattern}
\end{figure*}

Clearly the congestion of this pattern is uniformly $1$. Our algorithm
will use $\log n + 1$ wavelengths, each serving all transmissions of a
single class.

Consider an optimal solution for the non-splittable case. There has
to be a light-trail in which the transmission $[0, n]$ is scheduled.
Thus we must have a wavelength with no OFF shutters except at node 0
and node $n$. In this wavelength, it is easily seen that the longest
transmissions should be scheduled. So we start assigning transmissions
of first few classes in this light-trail. Suppose, all the transmissions
for first $l$ classes are assigned. Then we have total $1 + 2 + 4 +
\cdots + 2^l=2^{l+1}-1$ transmissions assigned to this light-trail. Total
bandwidth requirement of these transmissions should be at most~1. This
gives us $(2^{l+1}-1)(1/(\log n +1)) \leq 1$ implying $l\leq\log(\log
n+2)-1 \approx \log \log n$.

For the subsequent classes of transmissions, we allocate a new wavelength
and create light-trails by putting shutters OFF at nodes $qn/2^{l+1}$
for all possible $q$. It can be seen that again transmissions of next
about $\log \log n$ classes can be put in these light-trails. We repeat
this process until all transmissions are assigned.

In each wavelength we assign transmissions of $\log \log n$
classes. There are total $(1+\log n)$ classes. Thus the total
number of wavelengths needed is $\lceil(1+\log n)/\log \log n\rceil=
\Omega(\log{n}/\log\log{n})$ rather than the congestion bound of~$1$.

For the example in Fig.~\ref{fig:lbpattern}(a), using this procedure,
we have $\log \log n =2$. Thus we require $\lceil(1+\log n)/\log
\log n \rceil=3$ wavelengths. The first wavelength is used for
the transmissions of classes~$\{3,4\}$, the second wavelength for
classes~$\{1,2\}$ and the third for class~$0$. This solution is shown
in Fig.~\ref{fig:lbpattern}(b).

Suppose we add $c-1$ transmissions of extent $[0,n]$ and bandwidth
requirement $1$ to this pattern of transmissions of congestion $1$
uniformly at all links.  We can similarly show that an optimal solution
will require $c-1+\Omega(\log{n}/\log\log{n})$ wavelengths. Thus we
will get an instance of congestion $c$ uniformly at all links but which
requires $\Omega(c+\log{n}/\log\log{n})$ wavelengths, for any $c$.

\section{The On-line Problem}
\label{sec:online}

In the on-line case, the transmissions arrive dynamically.  An arrival
event has parameters $(s_i,d_i,r_i)$ respectively giving the origin,
destination, and bandwidth requirement of an arriving transmission
request.  The algorithm must assign such a transmission to a light-trail
$L$ such that $s_i,d_i$ belong to the light-trail, and at any time the
total bandwidth requirement of transmissions assigned to any light-trail
is at most $1$.  A departure event marks the completion of a previously
scheduled transmission.  The corresponding bandwidth is released and
becomes available for future transmissions.  The algorithm must make
assignments without knowing about subsequent events.

Unlike the stationary problem, congestion at any link may change over
time. Let $c_{lt}(S)$ denote the congestion induced on a link $l$ at time
$t$ by a set of transmissions $S$.  This is simply the total bandwidth
requirement of those transmissions from $S$ requiring to cross link $l$ at
time $t$. The congestion lower bound $c(S)$ is $\max_l \max_t c_{lt}(S)$,
the maximum congestion over all links over all time instants.

For the on-line problem, we present two algorithms -- (1) \sepc having
competitive ratio $\Theta(\log{n})$ and (2) \alc, a simplification
of \sepc inspired by the analysis of the algorithm for the stationary
problem, as may be seen. However, we show that this simplification not
necessarily makes it better. In fact we show that \alc has competitive
ratio in between $\Omega(\log^2 n/\log \log n)$ and $O(\log^2{n})$.


Now we present our two on-line algorithms. In both the on-line algorithms,
when a transmission request arrives, we first determine its class
$i$ and trail-point $x$ (defined in Section~\ref{ssec:conquer}).
The transmission is allocated to some light-trail $[x, x+2^{i+1}]$.
However, the algorithms differ in the way a light-trail is configured
on some candidate wavelength.

\subsection{Algorithm \sepc}
\label{sec:FTSC}

In this algorithm, every allocated wavelength is assigned a class label
$i$ and a phase label $p$, and has shutters OFF at nodes $(4q+p)2^{i-1}$
for all $q$, i.e., is configured to serve only transmissions of class
$i$ and phase $p$.  Whenever a transmission $t$ of class $i$ and phase
$p$ is to be served, it is only served by a wavelength with the same
labels.  If such a wavelength $w$ is found, and a light-trail $L$ on $w$
starting at the trail-point of $t$ has space, then $t$ is assigned to the
light-trail $L$.  If no such wavelength is found, then a new wavelength
$w'$ is allocated, it is labeled and configured for class $i$ and phase
$p$ as above and $t$ is assigned to the light-trail on $w'$ that starts
at the trail-point of $t$.

When a transmission finishes, it is removed from its associated
light-trail.  When all transmissions in a wavelength finish, then its
labels are removed, and it can subsequently be used for other classes
or phases.

\begin{lemma}
Suppose, at some instant of time, among the wavelengths allocated by \sepc,
$x$ wavelengths had non-empty light-trails of the same class and phase
across a link $l'$.  Then there must be a link $l$ having congestion
$\Omega(x)$ at some instant of time.  \label{lem:online1}
\end{lemma}

\begin{proof}
Suppose at some instant of time, wavelengths $w_1, w_2, \ldots, w_x$,
ordered according to the time of allocation, had non-empty light-trails
$L_1, L_2, \ldots, L_x$, respectively, of same class and phase across
link $l'$.  Let $u$ be the anchor (defined in Section~\ref{ssec:conquer})
of the transmissions assigned on these light-trails and $l$ be the link
between node $u$ and node $u+1$.

Now suppose wavelength $w_x$ was allocated due to a transmission $t$.
This could only happen because $t$ could not fit in the wavelengths $w_j$
for all $j\le x-1$.

For the splittable case this can only happen if light-trails $L_1$
through $L_{x-1}$ together contain transmissions of congestion at least
$x-1-c(t) = \Omega(x)$ crossing the anchor $u$ of $t$, when $t$ arrived.
Thus at that time $l$ had congestion $\Omega(x)$, giving us the result.

For the non-splittable case, suppose that $c(t)\le 0.5$.  Then the
transmissions in each of the light-trails $L_j, 1\le j \le x-1$, must
have congestion of least $0.5$ at $l$ when $t$ arrived, giving congestion
$\Omega(x)$.  So suppose $c(t)>0.5$.  Let $k$ be the largest such that
light-trail $L_k$ contains a transmission $t'$ with $c(t')\le 0.5$ when
$t$ arrived.  If no such $k$ exists, then clearly the congestion at $l$
when $t$ arrived is $\Omega(x)$.  If $k$ exists, then all the light-trails
$L_j$, $j > k$ have transmissions of congestion at least $0.5$ at $l$ when
$t$ arrived.  And the light-trails $L_j$, $j \le k$ had transmissions
of congestion at least $0.5$ at $l$ when $t'$ arrived.  So at one of
the two time instants the congestion at $l$ must have been $\Omega(x)$.
\end{proof}

\begin{theorem}
\sepc is $\Theta(\log n)$ competitive.
\label{thm:online1}
\end{theorem}

\begin{proof}
Suppose that \sepc uses $w$ wavelengths. We will show that the best
possible algorithm (including off-line algorithms) must use at least
$\Omega(w/\log{n})$ wavelengths. That will prove that \sepc is $O(\log
n)$ competitive.

Consider the time at which the $w$th wavelength was allocated.
At this time $w-1$ wavelengths are already in use, and of these at least
$w'=(w-1)/(4\log{n})$ must have the same class and phase. Among these $w'$
wavelengths consider the one which was allocated last to accommodate
some light-trail $L$ serving some newly arrived transmission. At that
time, each of the previously allocated $w'-1$ wavelengths was nonempty
in the extent of $L$.  By Lemma~\ref{lem:online1}, there is a link
that had congestion $\Omega(w'-1) = \Omega(((w-1)/(4\log{n}))-1) =
\Omega(w/\log{n})$ at some time instant. This is a lower bound on any
algorithm, even off-line.

We show the lower bound $\Omega(\log n)$ using the following example. Let
$n=2^k+1$. At each time $t=0,1,\ldots,k$, a transmission $[0, 2^t]$
arrives.  All transmissions have bandwidth requirement $1/(k+1)$. At
time $k+1$ all transmissions depart together. \sepc takes
$k$ wavelengths because each transmission is of a different class. The
optimal off-line algorithm assigns all of them to a single light-trail
spanning the complete network and hence takes only one wavelength.
\end{proof}

\subsection{Algorithm \alc}
\label{sec:FTAC}

This is a simplification of \sepc in that the allocated wavelengths are
not labeled.  When a transmission $t$ of class $i$ and trail-point $x$
arrives, we search the wavelengths in the order they were allocated for
a light-trail $L$ of extent $[x, x+2^{i+1}]$ such that $L$ has enough
space to serve $t$. If such a light-trail $L$ is found, then $t$ is
assigned to $L$.  If no such light-trail is found, then an attempt is
made to create a light-trail $[x, x+2^{i+1}]$ from the unused portions
of one of the existing wavelengths in a first-fit manner in the order
they were allocated. If such a light-trail $L$ can be created, then $L$
is created and $t$ is assigned to $L$.  Otherwise a new wavelength $w$
is allocated, the required light-trail $L$ of extent $[x, x+2^{i+1}]$ is
created on $w$, and $t$ is assigned to $L$. The portion of the wavelength
$w$ outside the extent of $L$ is marked unused.

When a transmission finishes, it is removed from its associated
light-trail. If this makes the light-trail empty then we mark its extent
on the corresponding wavelength as unused.

\begin{theorem}
\alc is $O(\log^2n)$ competitive.
\end{theorem}

\begin{proof}
Suppose \alc uses $w$ wavelengths.  We will show that an
optimal algorithm will use at least $\Omega(w/\log^2n)$.  Clearly,
we may assume $w=\Omega(\log^2n)$.

We first prove that there must exist a point of time in the execution
of \alc when there are at least $w/(4\log{n})$ non-empty
light-trails (not necessarily of same class and phase) crossing the
same link.

Number the wavelengths in the order of allocation.  Consider the
transmission $t$ for which the $w$th wavelength was allocated for the
first time.  Let $L$ be the light-trail used for $t$.  Clearly, the
$w$th wavelength had to be allocated because at that time the $w-1$
previously allocated wavelengths contained light-trails overlapping
with $L$. Let $S'$ denote this set of light-trails overlapping with
$L$. A light-trail $L' \in S'$ can be any of the three types - (1) $L$
and $L'$ overlap at the leftmost link of $L$, (2) $L$ and $L'$ overlap
at the rightmost link of $L$ but not at the leftmost link of $L$, and
(3) $L'$ is totally contained in the extent of $L$, without containing
either the leftmost or the rightmost link of $L$. Construct $S \subseteq
S'$ having exactly one light-trail from each of the $w-1$ wavelengths
by picking light-trails from $S'$ in the order -- first of type 1, then
of type 2 and finally of type 3. Now we consider three possible cases.

Case 1: $S$ has at least $w/(4\log{n})$ light-trails of type 1. Then
we have $w/(4\log{n})$ light-trails crossing the leftmost link of $L$.

Case 2: $S$ has at least $w/(4\log{n})$ light-trails of type 2. Then we
have $w/(4\log{n})$ light-trails crossing the rightmost link of $L$.

Case 3: $S$ has less than $w/(4\log{n})$ light-trails of
type 1 and less than  $w/(4\log{n})$ light-trails of type
2. Then, number of light-trails of type 3 in $S$ must be at least
$w'=w-1-2w/(4\log{n})=w-1-w/(2\log{n})$. Let $L'\in S$ be the light-trail
allocated on the $w'$th of the wavelengths having a light-trail of
type 3 in $S$. Note that $L'$ is strictly smaller than $L$.  Thus we
can repeat the above argument by using $L'$ and $w'$ in place of
$L$ and $w$ respectively, only $\log{n}$ times, and if we fail each
time to find at least $w/(4\log{n})$ light-trails crossing a link,
we will end up with a light-trail $L''$ such that there are at least
$w''$ wavelengths having light-trails conflicting with $L''$, where
$w''=w-\log{n}-\log{n}(w/(2\log{n}))=w/2 - \log{n}\ge w/(4\log{n})$
for $w=\Omega(\log^2n)$.  But $L''$ is a single link and so we are done.

Of these $w/(4\log{n})$ light-trails, at least $w/(16\log^2n)$
must have the same class and phase.  But it can be shown that
Lemma~\ref{lem:online1} is also true for \alc, and hence there is a
link having congestion $\Omega(w/(16\log^2n))$ at some time instant.
But this is a lower bound on the number of wavelengths required by any
algorithm, including an off-line algorithm.
\end{proof}

\subsection{Lower Bound for \alc}

We give a sequence of transmissions for which \alc takes
$\Omega(\log^2 n/\log \log n)$ wavelengths but an optimal off-line
algorithm, \opt, requires only one wavelength.

\begin{theorem}
\alc is $\Omega(\log^2n/\log \log n)$ competitive.
\end{theorem}

Our transmission sequence consists of several (exact count will be
shown later) subsequences, which we call stages.  In all stages, all
transmissions have height, i.e., bandwidth requirement of $1/n^2$. Our
transmission sequence is such that, at any point of time, there are less
than $n^2$ active transmissions. \opt will put all transmissions in a
single light-trail using full length of a wavelength.  On the other hand,
it will be seen that \alc will allocate $\Omega(\log^2n/\log \log n)$
wavelengths in total for all stages.  We describe the first stage only,
the other stages are scaled versions of the first stage. The goal of the
first stage is to force \alc to allocate wavelengths with light-trail
patterns given in the following lemma.

\begin{lemma}
Let the network have $q+1$ nodes numbered $0,\ldots,q$. Then there is
a transmission sequence for which \alc allocates $k=\floor{\log q}$
wavelengths numbered $0,\ldots,k-1$, with the following pattern $P$
of light-trails: each wavelength $i$ has $\floor{n/k}$ unit-length
light-trails $[jk+i,jk+i+1]$, each containing a single transmission,
for all $j=0,\ldots,\floor{n/k}-1$.

\label{lem:phase1}
\end{lemma}

\begin{proof}
For simplicity we assume $q=2^k$, i.e., $k$ is exactly equal to $\log
q$. The general case can be similarly proved.

Note that the pattern $P$ has non-overlapping and unit-length
light-trails. We first describe how to create a unit-length light-trail
$[x,x+1]$ on any wavelength $i$. We will repeatedly use this procedure
to create the pattern $P$. For this purpose, we define $Hill(h,x)$
to be an ordered sequence of $h$ transmissions as follows. For each
$i=0,1,\ldots,h-1$, $Hill(h,x)$ contains a transmission that uses
the link $[x,x+1]$ and has class $k-1-i$, and some suitable phase.
The key point is that all the transmissions in a hill overlap but
have different classes, and hence \alc must assign them in distinct
light-trails on different wavelengths.  Thus starting from scratch,
the arrival of the transmissions in a hill will cause $h$ wavelengths
to be allocated.  For example, we show $Hill(h=4,31)$ on right half
of Fig.~\ref{fig:allclass_lbpattern}\subref{parta}.  Further, if a new
transmission $[x,x+1]$ arrives, it will cause one more wavelength to be
allocated. From now on, by creating (deleting) a hill we mean the arrival
(departure) of transmissions in a hill.

Now we describe how to generate the pattern $P$ using several hills. The
idea is to build the portion of the pattern on one wavelength at a
time from top to bottom, i.e., first create all light-trails of $P$
on wavelength $k-1$, then all light-trails on wavelength $k-2$ and so
on. Each unit-length light-trail $[x,x+1]$ on an wavelength is created
by temporarily creating an appropriate hill underneath it, then creating
the transmission $[x,x+1]$ and finally deleting the temporary hill.
However, to make sure that all the wavelengths numbered $0,\ldots,k-2$
that were allocated when the part of the pattern on wavelength $k-1$
was created, remain alive throughout, we use the following trick. We
create left half of $P$ first. Then we create the right half.

Before creating the left half, we first create hill $H=Hill(k-1,
q-1)$. This hill will survive until the left half is completely
created. Its sole purpose is to keep $k-1$ wavelengths alive. The left
half is created top to bottom as given in Algorithm~\ref{alg:createP}.
Consider the first execution of the insertion marked as belonging to $P$.
Because of the hill $H'$, this transmission will clearly be assigned to a
light-trail on wavelength $i$.  Note further that when transmissions in
$H'$ depart, the wavelengths $0,\ldots,i-1$ do not become empty because
of the presence of hill $H$ in the right half.  Thus the subsequent
iterations also force the transmissions to be assigned in wavelength $i$,
and so on.

\begin{algorithm}[h!tb]
 \begin{algorithmic}
  \FOR{$i = k-1$ {\bf downto} 0}
    \FOR{$j = 0$ \TO $q/2k$}
      \STATE Create a hill $H'=Hill(i,jk+i)$
      \STATE Insert an arrival event for transmission $[jk+i,jk+i+1]$
      \COMMENT{belongs to $P$}
      \STATE Remove hill $H'$
    \ENDFOR
  \ENDFOR
 \end{algorithmic}
\caption{Create left half of $P$} \label{alg:createP}
\end{algorithm}

At the end of the above, we will have created a pattern as shown in
Fig.~\ref{fig:allclass_lbpattern}\subref{parta}.  Since each light-trail
contains only one transmission, we just show the transmissions instead
of the light-trails.

Next we remove $H$, and execute the same code to create the right half of
$P$ on the $k$ wavelengths already allocated. Note that the light-trails
created in the left half now serve the purpose that $H$ did earlier.
At this point we will have the complete pattern $P$.
\end{proof}

The first stage is created by using Lemma~\ref{lem:phase1} with $q=n$. In
the second stage, we can treat every $k=\floor{\log n}$ nodes as a
single node, and think of the network as having $n'=\floor{n/k}$ nodes.
We create a pattern of height $\floor{\log{n'}}$ similar to $P$ but
with light-trails of length $k$ using Lemma~\ref{lem:phase1} with
$q=n'$. Since these light-trails are longer than the light-trails in
the previous stage, we can stack up the new pattern on the top of the
previous pattern.  We can keep doing this until $n'$ becomes less than
$2$. It can be shown that the number of stages is $\Omega(\log_{\log n}
n) = \Omega(\log n/\log \log n)$.

Let $T(n)$ denote the total height of the patterns thus created for $n+1$
nodes, then $T(n)$ is computed using the following recurrence:
\begin{equation}
 T(n) = \floor{\log n} + T(\floor{n/\floor{\log n}}) \qquad
 \text{or simply} \qquad T(n) = \log n + T(n/\log n)
\end{equation}
with the base condition $T(n)=0$ for $n \le 1$. It can be shown that
the recurrence has solution $T(n) = \Omega(\log^2 n/\log \log n)$.

Thus \alc will use $\Omega(\log^2 n/\log \log n)$ wavelengths for the
patterns created. Fig.~\ref{fig:allclass_lbpattern}\subref{partb} shows
all the transmissions active at the end of all stages, for the example
considered in Fig.~\ref{fig:allclass_lbpattern}\subref{parta}.

\begin{figure*}[h!tb]
\centering
\subfloat[ ][Pattern created midway of the first stage] {
\includegraphics[width=\textwidth]{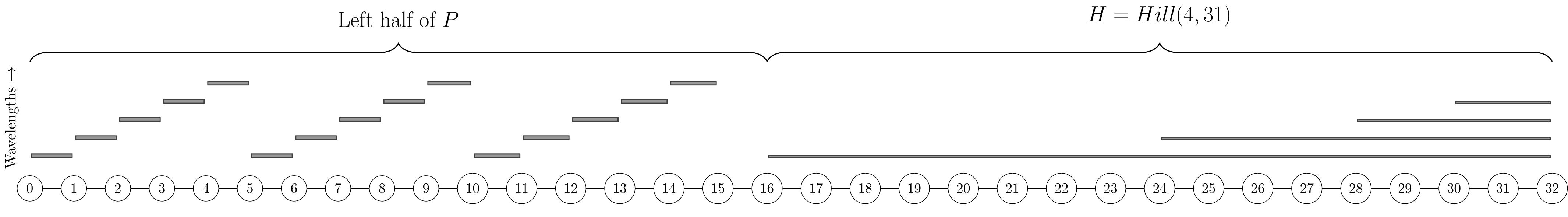}
\label{parta}
} \\
\subfloat[ ][Final Pattern] {
\includegraphics[width=\textwidth]{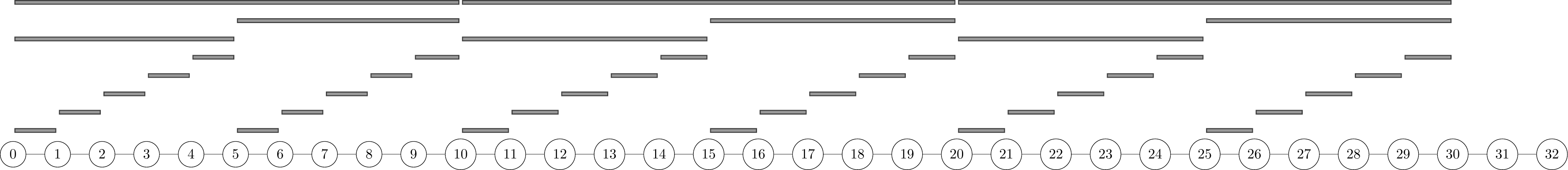}
\label{partb}
}
\caption{An instance for which \alc is $\Omega(\log^2 n/\log \log n)$ competitive}
\label{fig:allclass_lbpattern}
\end{figure*}

\subsection{Remarks}

It is interesting to note that \alc is more flexible than
\sepc, and it is this flexibility that is exploited in
the lower bound argument to show a worse ratio for \alc than
\sepc.

Indeed, the more flexibility we give, the worse it seems the ratio
will become.  In \alc, if we have a transmission of length $L=2^k$ we
assign it to a light-trail of length $2L$. This seems wasteful. But this
is done to accommodate transmissions that do not start at a multiple of
$L/4$, using only $4$ phases. Suppose we decide to be more flexible, and
allow light-trails to start anywhere (so long as their length is $2^k$
for some $k$) using $2^k$ phases.  Although this strategy will handle the
above transmission better, in general it is worse in that its competitive
ratio can be shown to be $\Omega(\log^2n)$.  We omit the details.

\section{Problem Lower Bound -- $\Omega(\log n)$}
\label{sec:problemlb}

Let \alg be some algorithm for the on-line problem and \opt be an optimal
off-line algorithm. By observing the behaviour of \alg we can create a
sequence of transmissions for which \alg takes $\Omega(\log{n})$ times
as many wavelengths as \opt.  This gives us the lower bound.

\begin{theorem}
There exists a transmission sequence for which any on-line algorithm
requires $\Omega(\log{n})$ times as many wavelengths as the optimal
off-line algorithm.
\label{thm:online_lb}
\end{theorem}

For convenience we assume that the network has $n+1$ nodes numbered $0, 1,
\ldots, n$ and $n=2^k$ for some $k$. The transmission sequence we use to
prove Theorem~\ref{thm:online_lb} is made of $k=\log n$ stages.  For stage
$i =0,1, \ldots, k-1$, consider the network broken up into $n'=n/l$
intervals of length $l=2^i$. Let this set of intervals be $Q_i =
\{[ql, (q+1)l]\}_{q = 0}^{n'-1}$.  At the beginning of the $i$th stage,
for each interval $I \in Q_i$, $k^2$ transmissions having extent
$I$ arrive.  We will denote this set of $n'k^2$ transmissions by $A_i$.
All transmissions have height (i.e., bandwidth requirement) of $1/k$.
At the end of the $i$th stage, all but a subset $S_i$ of $A_i$ depart.
The set $S_i$ is found with properties as per the following Lemma.

\begin{lemma}
Among the $n'k^2$ transmissions arriving at the beginning of stage $i$,
we can find a set $S_i$ of $n'k$ transmissions such that (1) Exactly $k$
transmissions from $S_i$ are for a single interval $I \in Q_i$,
(2) \alg assigns each transmission in $S_i$ to a distinct light-trail.
\label{matchlemma}
\end{lemma}

\begin{proof}
We have $k^2$ transmissions for each interval $I \in Q_i$.
Partition these $k^2$ transmissions arbitrarily into $k$ groups of $k$
transmissions each. So overall we have $n'k$ groups each containing $k$
transmissions. Now form a bipartite graph $(U,V,E)$ as follows.
\begin{enumerate}
\item $U$ has $n'k$ vertices, each vertex corresponding to a group of $k$
transmissions as formed above.  Note that there are $k$ groups for each
interval, and hence we can consider a distinct group of $k$ vertices of
$U$ to be associated with each $I \in Q_i$.
\item $V$ has a vertex corresponding to each light-trail used by
\alg for serving the transmissions of this stage.
\item $E$ has following edges. Suppose a transmission $t$ from the group
associated with a vertex $u\in U$ is placed by \alg in the light-trail
$L$ associated with a vertex $v \in V$.  Then for each such $t$ there
will be an edge $(u,v)$ in $E$.  Note that this may produce parallel
edges if several transmissions in the group of $u$ are placed in $L$.
\end{enumerate}
The degree of each vertex in $U$ is exactly $k$, one edge for each
transmission in the associated group.  Consider any vertex $v\in V$. Since
its associated light-trail can accommodate at most $k$ transmissions of
height $1/k$, its degree must be at most $k$.

Now consider any subset $S$ of vertices from $U$ and its neighborhood $T$
in $V$.  Because vertices in $U$ have degree exactly $k$ there must be
exactly $|S|k$ edges leaving $S$.  These must be a subset of the edges
entering $T$. But vertices in $V$ have degree at most $k$. So there
can be at most $|T|k$ edges entering $T$.  Thus we have $|S|k\le |T|k$,
i.e., $|T|\ge |S|$, i.e., $S$ has at least as many neighbors as its own
cardinality.  But this is true for any $S$.  Thus by the generalization
of Hall's theorem, there must be a matching $M$ that includes an edge
from every vertex of $U$ to a distinct vertex in $V$.

Consider the set $S_i$ of transmissions associated with each edge of $M$.
Since there is exactly one edge in $M$ for each node in $U$, $S_i$ has
one transmission per group of transmissions for each interval. Hence
$S_i$ has exactly $k$ transmissions for each interval.  Since $M$ has
exactly one edge per vertex in $V$, we know that each transmission in
$S_i$ is assigned to a distinct light-trail by \alg.
\end{proof}

We have now completely described the transmission sequence.  At the end all
transmissions have departed except those in some $S_i$.  We will use $D_i$
to denote the transmissions which depart in stage $i$.  Clearly
$A_i=S_i\cup D_i$.

\begin{lemma}
\opt uses overall $2k-1$ wavelengths while processing the
transmission sequence for all stages.
\end{lemma}

\begin{proof}
Consider stage $i$.  The set $A_i$ has $k^2$ transmissions for each
interval $I \in Q_i$.  To serve these transmissions $A_i$, \opt uses $k$
wavelengths configured as follows.  Each wavelength is configured into
light-trails as per $Q_i$, i.e., each interval $I \in Q_i$ forms one
light-trail.  Now the key point is that \opt places all transmissions
in $S_i$ into light-trails on a single wavelength.  This can be done
because the set $S_i$ indeed has $k$ transmissions for each $I\in Q_i$.
The remaining transmissions $D_i$ can be accommodated into $k-1$
additional wavelengths.  Note now that at the end of the stage,
the transmissions $D_i$ depart.  Hence although the stage used $k$
wavelengths transiently, at the end $k-1$ of these are released.

Thus, at the end of stage $i$, there will be $i+1$ wavelengths in use,
one for transmissions in each $S_j$, $j=0,\ldots,i$.  When $A_{i+1}$
arrives, \opt will allocate $k$ new wavelengths.  So while processing
$A_{i+1}$ there will be $i+1+k$ wavelengths in use.  These will drop
down to $i+2$ at the end of stage $i+1$.  Thus, over all the stages the
maximum number of wavelengths used will be at most $\max_{i=0,\ldots,k-1}
(i+k)$, i.e. $2k-1$.
\end{proof}

\begin{lemma}
\alg uses at least $k^2/2$ wavelengths while processing the
transmission sequence.
\end{lemma}

\begin{proof}
Consider the light-trails used by \alg which are active at the end of the
stage $k-1$.  Each of these light-trails may contain several transmissions
but only one transmission from each $S_i$.  Since transmissions from each
$S_i$ have different lengths, each light-trail must hold transmissions
of different lengths.  Thus, each light-trail can have at most one
transmission of length $1$, one of length 2, and so on.  The sum of the
lengths of the transmissions assigned to a single light-trail of length
$l$ is thus at most $1+2+4+\dots+l=2l-1$.  In other words, if we think
each light-trail of length $l$ has capacity of $l \times 1 =l$ units
and each transmission of length $l'$ uses $l'\times 1/k=l'/k$ units then
the light-trail is being used to an efficiency of $((2l-1)/k)/l \le 2/k$.

In general, a single wavelength, however it is partitioned into
light-trails, can accommodate light-trails of total capacity $n \times
1 =n$ units.  If each light-trail is used to efficiency only $2/k$,
then each wavelength can hold transmissions of total capacity at most $n
\times 2/k = 2n/k$. Since each transmission of unit length uses $1/k$
unit capacity, each wavelength can hold transmissions of total length
$(2n/k)/(1/k)=2n$. However, the transmissions that survive at the end
consist of $nk$ transmissions of length $1$, $nk/2$ transmissions of
length $2$, and so on to $2k$ transmissions of length $n/2$.  Thus the
total length is $nk^2$.  Thus \alg needs at least $(nk^2)/(2n)=k^2/2$
wavelengths at the end.
\end{proof}

But the maximum number of wavelengths needed by \opt is $2k-1$,
hence the competitive ratio is at least $k/4 = \Omega(\log n)$. This
completes the proof of Theorem~\ref{thm:online_lb}.


\section{Simulations}
\label{sec:simul}

We simulate our two on-line algorithms and a baseline algorithm on a
pair of oppositely directed rings, with nodes numbered 0 through $n-1$
clockwise.

We use slightly simplified versions of the algorithms described in
Section~\ref{sec:online} (but easily seen to have the same bounds):
basically we only use phases $0$ and $2$.  Any transmissions that would
go into class $i$ phase $1$ (or phase $3$) light-trail are contained
in some class $i+1$ light-trail (of phase $0$ or $2$ only), and are put
there.  We define a class $i$ and phase $0$ light-trail to be one that
is created by putting OFF shutters at nodes $jn/2^i$ for different $j$,
suitably rounding when $n$ is not a power of $2$. A light-trail with
class $i$ and phase $2$ is created by putting OFF shutters at nodes
$(jn/2^i + n/2^{i+1})$, again rounding suitably. The class and phase
of a transmission is determined by the light-trail of maximum class
(note that now larger classes have shorter light-trails) and minimum
phase that can completely accommodate it. For \alc, there
is a similar simplification.  Basically, we use light-trails having
end nodes at $jn/2^i$ and $(j+1)n/2^i$ or at $jn/2^i + n/2^{i+1}$ and
$(j+1)n/2^i + n/2^{i+1}$.  As before, in \sepc, we require
any wavelength to contain light-trails of only one class and phase;
whereas in \alc, a wavelength may contain light-trails of
different classes and phases.

For the baseline algorithm in each ring we use a single OFF shutter
at node 0.  Transmissions from lower numbered nodes to higher numbered
nodes use the clockwise ring, and the others, the counterclockwise ring.

\subsection{The Simulation Experiment}
\label{sec:req_generation}

A single simulation experiment consists of running the algorithms on a
certain load, characterized by parameters $\lambda,D,r_{min}$ and $\alpha$
for $100$ time steps.  In our results, each data-point reported is the
average of $150$ simulation experiments with the same load parameters.

In each time step, all nodes $j$ that are not busy transmitting, generate
a transmission $(j,d_j,r_j)$ active for $t_j$ time units.  After that
the node is busy for $t_j$ steps.  After that it generates another
transmission as before.  The transmission duration $t_j$ is drawn from
a Poisson distribution with parameter $\lambda$. The destination $d_j$
of a transmission is picked using the distribution $D$ discussed later.
The bandwidth is drawn from a modified Pareto distribution with scale
parameter $=100 \times r_{min}$ and shape parameter $=\alpha$. The
modification is that if the generated bandwidth requirement exceeds the
wavelength capacity $1$, it is capped at $1$.

We experimented with $\alpha=\{1.5,2,3\}$ and $\lambda=\{0.01,0.1\}$ but
report results for only $\alpha=1.5$ and $\lambda=0.01$; results for other
values are similar.  We tried four values $0.01, 0.1, 0.25$ and $0.5$
for $r_{min}$.
We considered four different distributions $D$ for selecting the
destination node of a transmission.
\begin{enumerate}
\item \emph{Uniform}: we select a destination uniformly randomly from
the $n-1$ nodes other than the source node.
\item \emph{UniformClass}: we first choose a class uniformly from the
$\lceil \log n/2 \rceil+1$ possible classes and then choose a destination
uniformly from the nodes possible for that class. It should be noted that
there can be a destination at a distance at most $n/2$ in any direction
since we schedule along the direction requiring shortest path.
\item \emph{Bimodal}: first we randomly choose one of two possible
modes. In mode 1, a destination from the two immediate neighbors is
selected and in mode 2, a destination from the nodes other than the two
immediate nodes is chosen uniformly. For applications where transmissions
are generated by structured algorithms, local traffic, i.e., unit or
short distances (e.g. $\sqrt{n}$ for mesh like communications) would
dominate. Here, for simplicity, we create a bimodal traffic which is
mixture of completely local and completely global.
\item \emph{ShortPreferred}: we select destinations at shorter distance
with higher probability. In fact, we first choose a class $i$ in the range
$0,\ldots,\lceil \log n/2 \rceil$ with probability $\frac{1}{2^{i+1}}$ and
then select a destination uniformly from the possible destinations in that
class.
\end{enumerate}

We report the results only for the distributions \emph{Uniform}
and \emph{Bimodal} and for $r_{min}=0.01,0.5$, i.e., total 4 load
scenarios. Results for other scenarios follow similar pattern.

\subsection{Results}

Fig.~\ref{fig:simres} shows the results for the 4 load scenarios.
For each scenario, we report the number of wavelengths required
by the 3 algorithms and the measured congestion as defined in
Section~\ref{sec:online}.  Each data-point is the average of $150$
simulations (each of $100$ time steps) for the same parameters on
rings having $n=5,6,\ldots,20$ nodes.  We say that the two scenarios
corresponding to $r_{min}=0.01$ have {\em low load} and the remaining
two scenarios ($r_{min}=0.5$) have {\em high load}.

\begin{figure*}[h!tb]
 \centering
 \subfloat[Low Load]{\label{fig:sim1}\includegraphics[width=\textwidth]{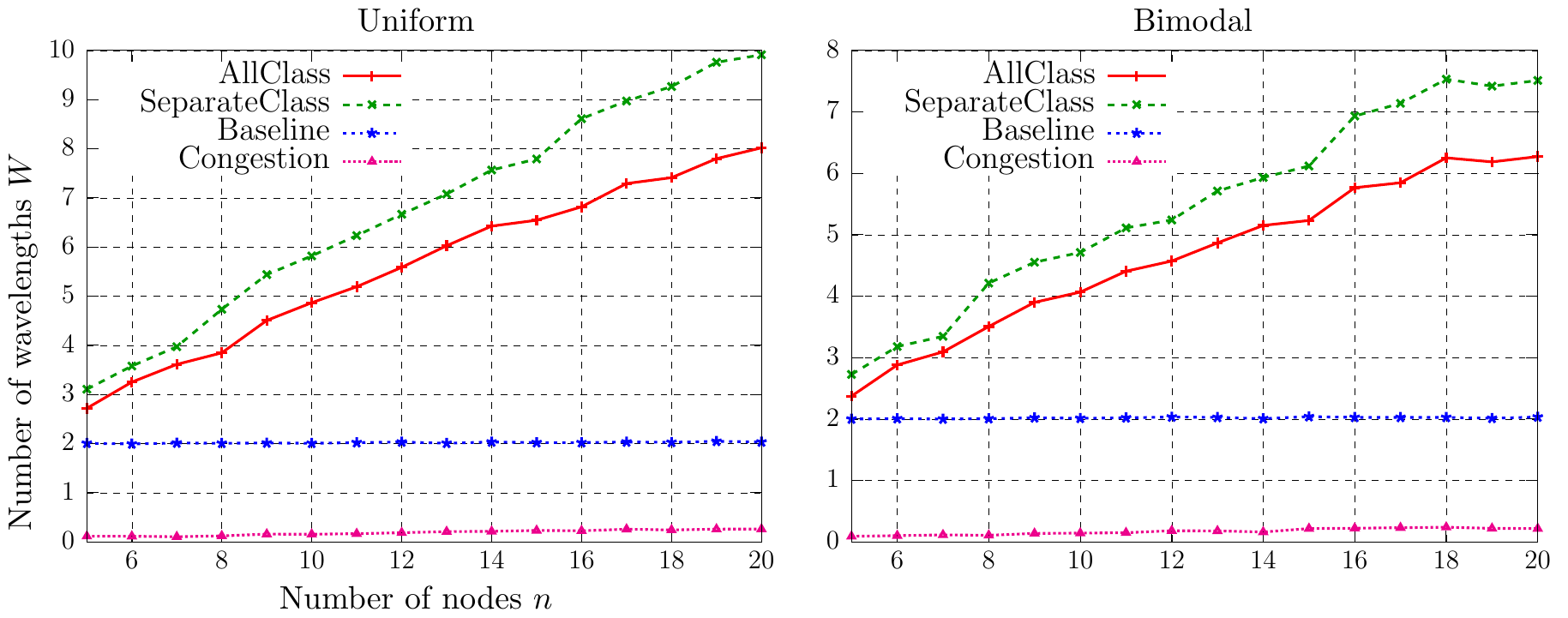}}\\
 \subfloat[High Load]{\label{fig:sim3}\includegraphics[width=\textwidth]{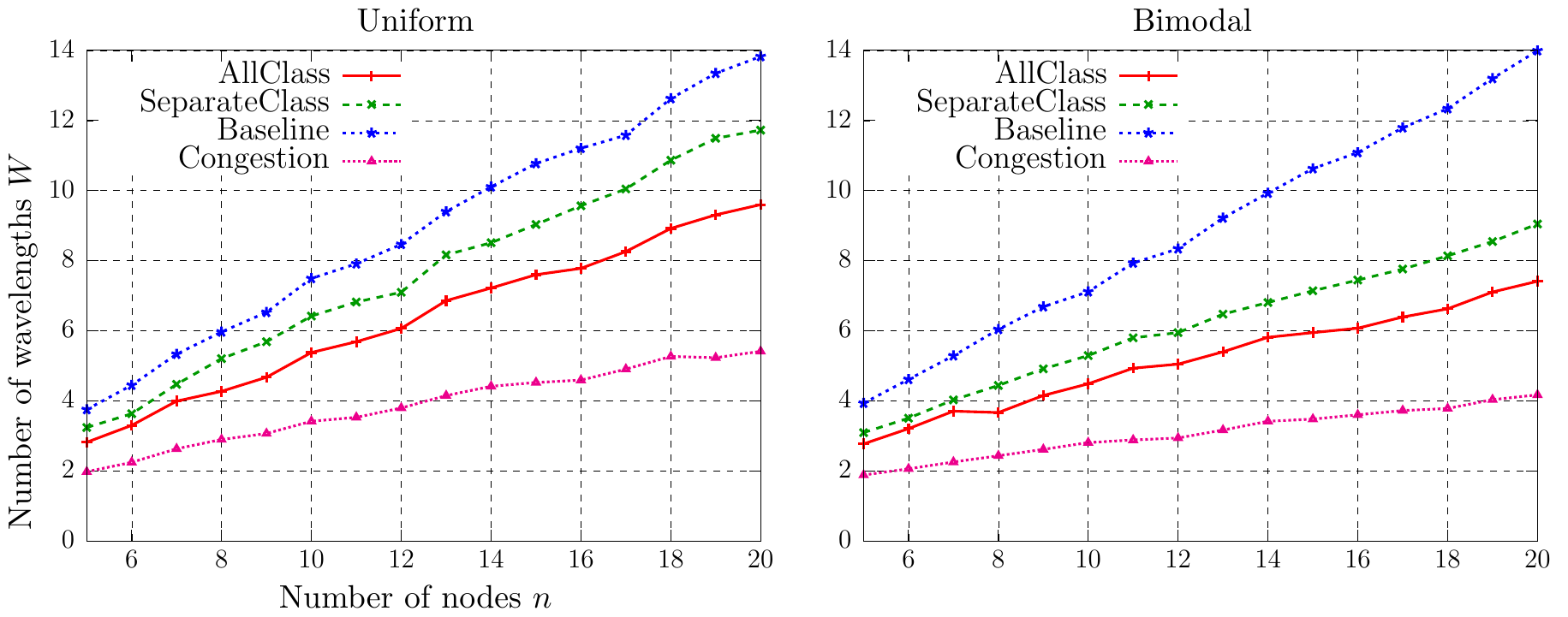}}\\
 \caption{Simulation results}
 \label{fig:simres}
\end{figure*}

For low load, the baseline algorithm outperforms our algorithms.
At this level of traffic, it does not make sense to reserve different
light-trails for different classes.  However, as load increases our
algorithms outperform the baseline algorithm.

For the same load, it is also seen that our algorithms become more
effective as we change from the completely global {\em Uniform}
distribution to the more local {\em Bimodal} distribution.  This trend
was also seen with the other distributions we experimented with.

Though we showed analytically that \alc is not better than
\sepc always, our simulation shows that \alc
performs better on the type of loads we generated.  It may be noted
that our algorithm for the stationary case mixes up the light-trails of
different classes, and so suggests that the \alc might work
better in many practical scenarios.

\section{Conclusions and Future Work}

It can be shown that the non-splittable stationary problem is
NP-hard, using a simple reduction from bin-packing.  We do not know
if the splittable problem is also NP-hard. We gave an algorithm
for both variations of the stationary problem which takes $O(c+\log
n)$ wavelengths.  It will also be useful to improve the lower bound
arguments; as Section~\ref{sec:lb} shows, congestion is not always a
good lower bound.  This may lead to a constant factor approximation
algorithm for the problem.

In the on-line case we proved that the lower bound on the competitive
ratio of any algorithm is $\Omega(\log n)$ and gave a matching algorithm
which we proved to have competitive ratio $\Theta(\log n)$. We also gave
a second algorithm which seems to work better in practice but can be
as bad as $\Omega(\log^2 n/\log \log n)$ factor worse than an optimal
off-line algorithm on some pathological examples as we have shown. We
also proved an upper bound of $O(\log^2 n)$ for the algorithm but it
will be an interesting problem to close the gap between the two bounds.

Our on-line model is very conservative in the sense that once a
transmission is allocated on a light-trail, the light-trail cannot
be modified. However, other models allow light-trails to shrink/grow
dynamically~\cite{gumaste2004lto}.  It will be useful to incorporate this
(with some suitable penalty, perhaps) into our model.

It will also be interesting to devise special algorithms that work well
given the distribution of arrivals.

\section*{Acknowledgment}

We would like to thank Ashwin Gumaste for encouragement, insightful
discussions and patient clearing of doubts related to light-trails.

\bibliographystyle{unsrt}
\bibliography{lighttrail,ip,graph}

\end{document}